\documentclass[a4paper]{article}
\usepackage[utf8]{inputenc}

\usepackage[margin=1.3in]{geometry}

\usepackage{amsthm,amssymb,amsfonts,amsmath}
\usepackage{amssymb,amsfonts,amsmath}
\usepackage{xcolor}
\newtheorem{theorem}{Theorem}
\newtheorem{lemma}[theorem]{Lemma}
\newtheorem{corollary}[theorem]{Corollary}

\newtheorem{observation}[theorem]{Observation}
\newtheorem{definition}{Definition}

\newtheorem{conjecture}{Conjecture}

\usepackage{graphicx}
\graphicspath{{figures/}}

\usepackage{thm-restate}
\usepackage{url,doi}
\usepackage{faktor}
\usepackage{tikz-cd}

\usepackage{dsfont}

\title{Combinatorial Depth Measures for Hyperplane Arrangements}
\author{Patrick Schnider \\ Department of Computer Science,
        ETH Z\"{u}rich, Switzerland \\ {\tt patrick.schnider@inf.ethz.ch}
        \and
        Pablo Sober\'{o}n \\ Department of Mathematics, Baruch College, City University of New York, USA \\
        Department of Mathematics, The Graduate Center, City University of New York, USA \\ {\tt psoberon@gc.cuny.edu}} 
\date{}

\begin{document}

\maketitle

\begin{abstract}
Regression depth, introduced by Rousseeuw and Hubert in 1999, is a notion that measures how good of a regression hyperplane a given query hyperplane is with respect to a set of data points. Under projective duality, this can be interpreted as a depth measure for query points with respect to an arrangement of data hyperplanes. The study of depth measures for query points with respect to a set of data points has a long history, and many such depth measures have natural counterparts in the setting of hyperplane arrangements. For example, regression depth is the counterpart of Tukey depth. Motivated by this, we study general families of depth measures for hyperplane arrangements and show that all of them must have a deep point. Along the way we prove a Tverberg-type theorem for hyperplane arrangements, giving a positive answer to a conjecture by Rousseeuw and Hubert from 1999. We also get three new proofs of the centerpoint theorem for regression depth, all of which are either stronger or more general than the original proof by Amenta, Bern, Eppstein, and Teng. Finally, we prove a version of the center transversal theorem for regression depth.
\end{abstract}

\section{Introduction}\label{sec:introduction}

%\patrick{This is just a vague sketch of the results. Feel free to change anything in the structure and notation. Also, if you can think of better names for the considered depth measures, feel free to change them :)}

A central topic in combinatorial geometry and computational geometry is the study of structural properties of finite families of points in Euclidean spaces.  Studying which sets can be separated from others by hyperplanes is a natural question, which leads us to study combinatorial properties of convex sets.  Classic results, such as Tverberg's theorem \cite{Tverberg:1966tb} and Rado's centerpoint theorem \cite{Rado:1946ud} follow from this line of thought.

In some cases, instead of being provided our data as a finite set of points in $\mathbb{R}^d$, we might receive it as a set of hyperplanes.  Understanding which results for families of points transfer to families of hyperplanes is a natural question.

Given a hyperplane arrangement $A$ in $\mathbb{R}^d$ and a point $q$, we first consider the depth of $q$ with respect to $A$ as follows.

\begin{definition}
The regression depth of a query point $q$ with respect to hyperplane arrangement $A$, denoted by $\text{RD(A,q)}$, is the minimum number of hyperplanes in $A$ intersected by or parallel to any ray emanating from $q$.
\end{definition}

Note that if $q$ lies on a hyperplane $H$, then any ray emanating from $q$ intersects $H$.  Regression depth has been widely studied \cite{Amenta, Fulek:2009vz, Rousseeuw1999, Rousseeuw1999a, Kreveld2008}.  In this manuscript we provide new structural results for regression depth, related to Tverberg's theorem and \textit{enclosing depth}.  In particular, given a finite arrangement $A$ of hyperplanes in $\mathbb{R}^d$, we might measure the depth of a point $q$ in $A$ in several different ways, so we study general properties of depth measures with respect to arrangements of hyperplanes.  This follows a similar approach recently taken for depth measures with respect to finite families of points \cite{Enclosing}.

Given an arrangement $A$ of $n$ hyperplanes, the existence of points with regression depth at least $n/(d+1)$ has been established by Amenta, Bern, Eppstein, and Teng \cite{Amenta}, and later by Mizera \cite{Mizera} as well as Karasev \cite{Karasev:2008jl}.  This can be considered a hyperplane version of Rado's centerpoint theorem \cite{Rado:1946ud}.  We give three new proofs of the existence of points with large regression depth.  First, we prove that the centerpoint theorem for regression depth is the consequence of a Tverberg-type theorem, confirming a conjecture of Rousseeuw and Hubert \cite{Rousseeuw1999}.  

\begin{restatable}{theorem}{tverberg}\label{thm:tverberg}
Let $r,d$ be positive integers and $A$ be an arrangement of at least $(r-1)(d+1)+1$ hyperplanes in $\mathbb{R}^d$.  Then, there exists a point $q$ in $\mathbb{R}^d$ and a partition of $A$ into $r$ parts such that $q$ has positive regression depth with respect to each of the $r$ parts.
\end{restatable}

This was previously known when $d=2$ \cite{Rousseeuw1999} or when $r$ is a prime power \cite{Karasev:2011jv, Karasev:2014ex}. The version for prime powers by Karasev holds with a slightly more restrictive version of regression depth. Based on this result, we define the \textit{hyperplane Tverberg depth} of a point.

%Let $A$ be an arrangement of hyperplanes in $\mathbb{R}^d$.
%An \emph{$r$-partition} of $A$ is a partition of $A$ into $r$ pairwise disjoint subsets $A_1,\ldots A_r\subseteq A$ such that there is a point $q$ for which $\text{RD}(A_i,q)\geq 1$ for all $i\in\{1,\ldots,r\}$.
%In this case we say that the intersection of the partition contains $q$.

\begin{definition}
The hyperplane Tverberg depth of a query point $q$ with respect to hyperplane arrangement $A$, denoted by $\text{HTvD(A,q)}$, is the maximum $r$ such that there is a partition of $A$ into $r$ parts such that $q$ has positive regression depth with respect to each part.
\end{definition}

Our other two proofs are topological, and each also has stronger consequences.  One proof based on a topological version of Helly's theorem shows the existence of points of high \textit{open regression depth}, which is a slightly weaker measure of depth introduced in Section \ref{sec:Helly}.  The last proof, based on properties of vector bundles, works for regression depth in families of weighted arrangements.

Another way to measure the depth of a point with respect to a hyperplane arrangement is via $k$-enclosures.  We say that an arrangement $A$ \emph{$k$-encloses} a query point $q$ if $A$ can be partitioned into $d+1$ pairwise disjoint subsets $A_1,\ldots,A_{d+1}$, each of size $k$, such that for every choice $h_1\in A_1,\ldots,h_{d+1}\in A_{d+1}$ we have that $\text{RD}(\{h_1,\ldots, h_{d+1}\},q)\geq 1$.

\begin{definition}
The hyperplane enclosing depth of a query point $q$ with respect to a hyperplane arrangement $A$, denoted by $\text{HED(A,q)}$, is the maximum $k$ such that there is a sub-arrangement of $A$ which $k$-encloses $q$.
\end{definition}

Given a finite hyperplane arrangement $A$, we prove the existence of points with high hyperplane enclosing depth with respect to $A$.  In particular, our lower bound is linear in $|A|$.  The existence of points with large enclosing depth for families of points has been established by Pach \cite{Pach:1998vx} and by Fabila-Monroy and Huemer \cite{FabilaMonroy2017} (see \cite{Enclosing} for improved constants).

One striking generalization of Rado's centerpoint theorem is the central transversal theorem, proven independently by Dolnikov and by \v{Z}ivaljevi\'c and Vre\'cica \cite{Dolnikov, Zivaljevic}.  In Section \ref{sec:center-transversal} we prove an analogue for hyperplane arrangements.  Given a hyperplane arrangement $A$ in $\mathbb{R}^d$ and a linear subspace $L$ in $\mathbb{R}^d$, we denote by $A \cap L$ the restriction of $A$ to $L$.  In Theorem \ref{thm:center-transversal}, we show that given $d-k+1$ different arrangements of hyperplanes in $\mathbb{R}^d$, there exists a $k$-dimensional linear subspace $L$ such that the restrictions of each arrangement to $L$ share a point with high regression depth.

In particular, just as the central transversal theorem generalizes the ham sandwich theorem, Theorem \ref{thm:center-transversal} has the following corollary.

\begin{corollary}
Let $A_1, \dots, A_d$ be $d$ hyperplane arrangements in $\mathbb{R}^d$.  There exists a line $\ell$ through the origin in $\mathbb{R}^d$ and a point $q \in \ell$ such that each of the two rays in $\ell$ starting from $q$ intersects at least $|A_i|/2$ hyperplanes of $A_i$, for each $i=1,\dots, d$.
\end{corollary}

The corollary above is similar to mass partition results for families of hyperplanes with segments \cite{Bereg:2015cz, RoldanPensado2022}, and to projective versions of the central transversal theorem \cite{Karasev:2014ex}.

\section{Correspondence to depth measures for point sets}
For an arrangement $A$ and a query point $q$, we define the \emph{dual of $A$ at $q$}, denoted by $A^*_q$, as follows.
For each hyperplane $h\in A$, let $p(h)$ be the unique point on $h$ that is closest to $q$.
We define $A^*_q$ as the set formed by all these points, that is, $A^*_q:=\{p(h)\mid h\in A\}$.
Note that if $q$ lies on $k$ hyperplanes, then those $k$ dual points coincide with $q$ in $A^*_q$.

Using this duality, for every depth measure $\rho$ on point sets we can define a corresponding depth measure $\rho^*$ on hyperplane arrangements and vice versa, by setting $\rho^*(A,q)=\rho(A^*_q,q)$.
We have the following observation.

\begin{observation}\label{obs:duality}
\begin{enumerate}
    \item a ray $r$ emanating from $q$ intersects a hyperplane $h$ if and only if the half-space $r^{\perp}$ defined by the hyperplane thorugh $q$ orthogonal to $r$, oriented such that it contains $r$, contains $p(h)$;
    \item the point $q$ has positive regression depth with respect to $h_1,\ldots,h_n$ if and only if it is in the convex hull of $p(h_1),\ldots,p(h_n)$.
    \item the point $q$ lies in the simplex defined by $h_1\ldots,h_{d+1}$ if and only if it is in the interior of the convex hull of $p(h_1),\ldots,p(h_{d+1})$.
\end{enumerate}
\end{observation}

The three depth measures for hyperplane arrangements defined in Section \ref{sec:introduction} all have natural corresponding depth measures for point sets that follow immediately from Observation \ref{obs:duality}. For regression depth, the corresponding depth measure is \emph{Tukey depth} ($\text{TD}$), which is defined as the minimum number of data points contained in any closed half-space containing the query point $q$ \cite{Tukey}. For hyperplane Tverberg depth we get Tverberg depth ($\text{TvD}$), which is defined as the maximum $r$ for which there exists an $r$-partition of the data points containing the query point $q$ in their intersection. Finally, for hyperplane enclosing depth, we get \emph{enclosing depth} ($\text{ED}$), which is defined as the maximum $k$ for which there exists a subset of the data points that $k$-encloses the query point $q$ \cite{Enclosing}.

\begin{corollary}\label{cor:depth_under_duality}
Let $A$ be an arrangement of hyperplanes in general position in $\mathbb{R}^d$ and let $q$ be a query point. Then
\begin{enumerate}
    \item $\text{RD}(A,q)=\text{TD}(A^*_q,q)$;
    \item $\text{HTvD}(A,q)=\text{TvD}(A^*_q,q)$;
    \item $\text{HED}(A,q)=\text{ED}(A^*_q,q)$.
\end{enumerate}
\end{corollary}

\section{Axioms for hyperplane depth}
Let $A^{\mathbb{R}^d}$ denote the family of all finite arrangements of hyperplanes in $\mathbb{R}^d$.
A depth measure for hyperplanes is a function $\rho:(A^{\mathbb{R}^d},\mathbb{R}^d)\rightarrow\mathbb{R}_{\geq 0}$ which assigns to each pair $(A,q)$ consisting of a hyperplane arrangement $A$ and a query point $q$ a value, which describes how deep the query point $q$ lies within the arrangement $A$. A depth measure is called \emph{combinatorial} if it is the same for all points in a face of $A$.
Similar to \cite{Enclosing}, we introduce some axioms, that reasonable depth measures for hyperplane arrangements should satisfy.

We say that a combinatorial depth measure for hyperplanes is \emph{super-additive} if it satisfies the following four conditions.

\begin{enumerate}
    \item[(i)] for all $A\in A^{\mathbb{R}^d}$ and $q\in\mathbb{R}^d$ and any hyperplane $h$ we have $|\rho(A,q)-\rho(A\cup\{h\},q)|\leq 1$,
    \item[(ii)] for all $A\in A^{\mathbb{R}^d}$ we have $\rho(A,q)=0$ if $q$ is in an unbounded cell of $A$,
    \item[(iii)] for all $A\in A^{\mathbb{R}^d}$ we have $\rho(A,q)\geq 1$ if $q$ is in a bounded cell or if $q$ lies on a hyperplane of $A$,
    \item[(iv)] for any disjoint subsets $A_1,A_2\subseteq A$ and $q\in\mathbb{R}^d$ we have $\rho(A,q)\geq \rho(A_1,q)+\rho(A_2,q)$.
\end{enumerate}

\begin{observation}
Regression depth and hyperplane Tverberg depth are super-additive, but hyperplane enclosing depth is not.
\end{observation}

For hyperplane enclosing depth, an example with $\text{HED}(A_1,q)=\text{HED}(A_2,q)=\text{HED}(A,q)=1$ can be found in Figure \ref{fig:encl_additivity}.

%\begin{proof}
%It follows straight from the definition that all three measures satisfy conditions (i) and (ii).
%As for condition (iii), note that any point in a bounded cell or on a hyperplane has regression depth at least 1, and thus also the other two measures evaluate to at least 1. It remains to analyze condition (iv).
%For regression depth, let $r$ be a ray witnessing $\text{RD}(A,q)=k$, that is, a ray intersecting (or parallel to) exactly $k$ hyperplanes of $A$.
%Of these hyperplanes, $k_1$ are in $A_1$ and $k_2$ are in $A_2$, where $k_1+k_2=k$. Thus, $r$ is a witness that $\text{RD}(A_i,q)\leq k_i$ for $k\in\{1,2\}$, and thus $\text{RD}(A,q)\geq\text{RD}(A_1,q)+\text{RD}(A_2,q)$.
%For hyperplane Tverberg depth, note that any $r_1$-partition of $A_1$ and any $r_2$-partition of $A_2$, both whose intersections contain $q$, define an $(r_1+r_2)$-partition of $A$ whose intersections still contains $q$.
%For hyperplane enclosing depth, an example with $\text{HED}(A_1,q)=\text{HED}(A_2,q)=\text{HED}(A,q)=1$ can be found in Figure \ref{fig:encl_additivity}.
%\end{proof}

\begin{figure}
\centering
\includegraphics[scale=0.6]{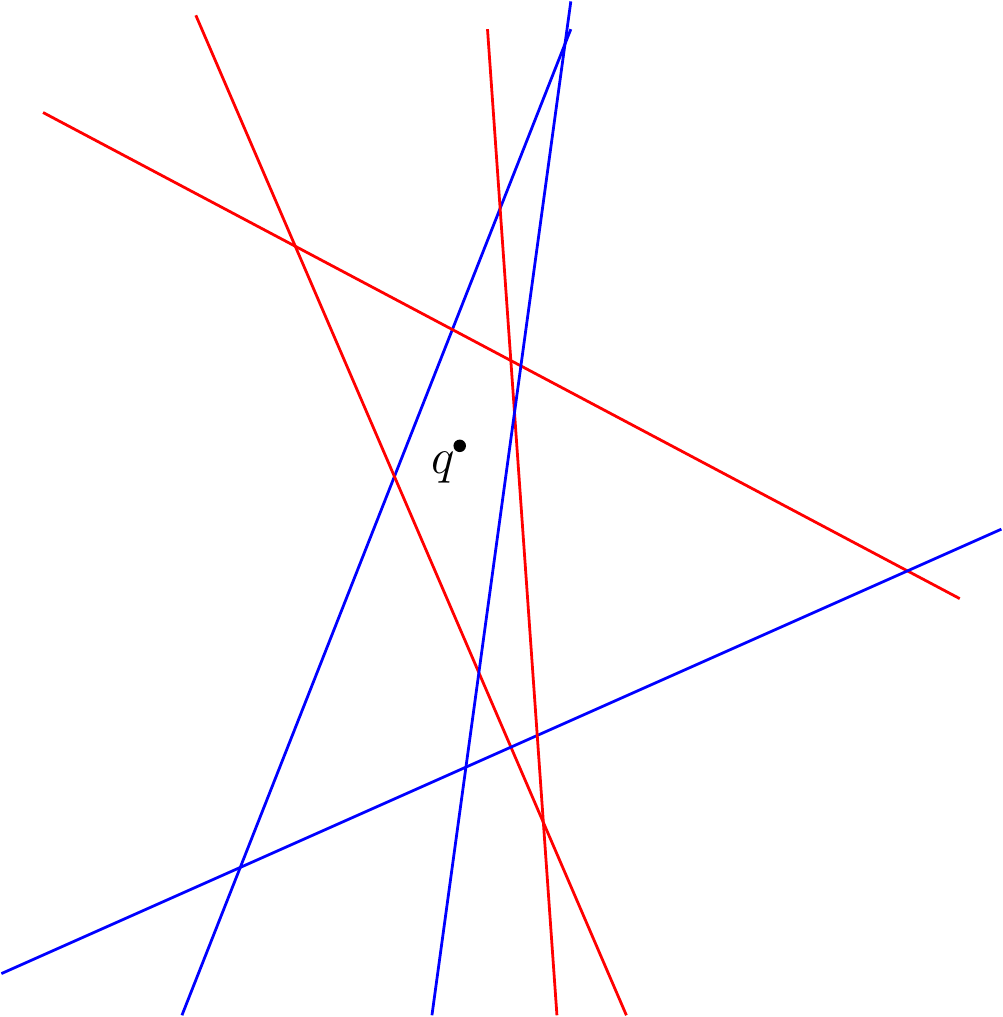}
\caption{Hyperplane enclosing depth does not satisfy condition (iv): the point $q$ has hyperplane enclosing depth 1 with respect to both the blue and the red lines, but its hyperplane enclosing depth with respect to the union of the two sets is still 1.}
\label{fig:encl_additivity}
\end{figure}

\begin{lemma}\label{lem:upper_bound}
Let $\rho$ be any combinatorial depth measure that satisfies conditions (i) and (ii). Then for all $A\in A^{\mathbb{R}^d}$ and $q\in\mathbb{R}^d$ we have $\rho(A,q)\leq\text{RD}(A,q)$.
\end{lemma}

\begin{proof}
Let $\text{RD}(A,q)=k$. This means that there is a ray $r$ which intersects or is parallel to some $k$ hyperplanes of $A$. Removing these $k$ hyperplanes, we get a new arrangement $A'$ and we have $\text{RD}(A',q)=0$. In particular, $q$ is in an unbounded cell of $A'$ and thus also $\rho(A',q)=0$ by condition (ii). By condition (i) we have $\rho(A,q)\leq \rho(A',q)+k=k$.
\end{proof}

\begin{lemma}
Let $\rho$ be any combinatorial depth measure that satisfies conditions (iii) and (iv). Then for all $A\in A^{\mathbb{R}^d}$ and $q\in\mathbb{R}^d$ we have $\rho(A,q)\geq\text{HTvD}(A,q)$.
\end{lemma}

\begin{proof}
Let $\text{HTvD}(A,q)=k$. This means that there is a $k$-partition $A_1,\ldots,A_k$ such that $q$ has regression depth $\geq 1$ with respect to each part. By condition (iii) we have $\rho(A_i,q)\geq 1$ for each $A_i$. By condition (iv) we get $\rho(A,q)\geq\rho(A_1,q)+\ldots +\rho(A_k,q)\geq k$.
\end{proof}

\begin{lemma}
For all $A\in A^{\mathbb{R}^d}$ and $q\in\mathbb{R}^d$ we have $\text{HTvD}(A,q)\geq\frac{1}{d}\text{RD}(A,q)$.
\end{lemma}

\begin{proof}
By Corollary \ref{cor:depth_under_duality} we have $\text{HTvD}(A,q)=\text{TvD}(A^*_q,q)$ and $\text{RD}(A,q)=\text{TD}(A^*_q,q)$. It is well known that for any point set $S$ in $\mathbb{R}^d$ and any query point $q$ we have $\text{TvD}(S,q)\geq\frac{1}{d}\text{TD}(S,q)$, see e.g.\ \cite{Amenta,Sariel,Rolnick}.
\end{proof}

Combining all of the above, we get

\begin{theorem}\label{thm:super-additive_measures}
Let $\rho$ be a super-additive depth measure for hyperplanes. Then for all $A\in A^{\mathbb{R}^d}$ and $q\in\mathbb{R}^d$ we have $\text{RD}(A,q)\geq\rho(A,q)\geq\text{HTvD}(A,q)\geq\frac{1}{d}\text{RD}(A,q)$.
\end{theorem}

As we have seen above, not all depth measures are super-additive: hyperplane enclosing depth is an example of a measure that is not. To include more general depth measures, we define a second family of measures, defined by a weaker set of axioms. We call a combinatorial depth measure for hyperplanes \emph{enclosable} if it satisfies the following conditions.

\begin{enumerate}
    \item[(i)] for all $A\in A^{\mathbb{R}^d}$ and $q\in\mathbb{R}^d$ and any hyperplane $h$ we have $|\rho(A,q)-\rho(A\cup\{h\},q)|\leq 1$,
    \item[(ii)] for all $A\in A^{\mathbb{R}^d}$ we have $\rho(A,q)=0$ if $q$ is in an unbounded cell of $A$,
    \item[(iii')] for all $A\in A^{\mathbb{R}^d}$ we have $\rho(A,q)\geq k$ if $A$ $k$-encloses $q$,
    \item[(iv')] for all $A\in A^{\mathbb{R}^d}$ and $q\in\mathbb{R}^d$ and any hyperplane $h$ we have $\rho(A\cup\{h\},q)\geq\rho(A,q)$.
\end{enumerate}

\begin{observation}
    Regression depth, hyperplane Tverberg depth and hyperplane enclosing depth are all enclosable.
\end{observation}

%\begin{proof}
%Conditions (i) and (ii) are the same as for super-additivity, and for all three measures it follows immediately from the definitions that they are satisfied. Condition (iii') is satisfied by definition for hyperplane enclosing depth, and condition (iv') follows from the fact that adding hyperplanes cannot destroy an enclosing subarrangement. For the other two measures, conditions (iii') and (iv') follow from the fact that they are super-additive. 
%\end{proof}

By Lemma \ref{lem:upper_bound}, any enclosable depth measure is bounded from above by regression depth.
On the other hand, it follows immediately from conditions (iii') and (iv') that any enclosable depth measure is bounded from below by hyperplane enclosing depth.
We finish this section by showing a lower bound for hyperplane enclosing depth.
In Theorem 17 in \cite{Enclosing} it was shown that there is a constant $c(d)$ such that for any point set $S$ in $\mathbb{R}^d$ and any query point $q$ we have $\text{ED}(S,q)\geq c\cdot \text{TD}(S,q)$.
Let now $q$ be a point of largest regression depth for a hyperplane arrangement $A$. We will see in Theorem \ref{thm:weighted_centerpoint} that $q$ has regression depth at least $\frac{|A|}{d+1}$. By Observation \ref{obs:duality}, this means $\text{TD}(A^*_q,q)\geq \frac{|A|}{d+1}$. By Theorem 17 in in \cite{Enclosing}, it follows that $\text{ED}(A^*_q,q)\geq \frac{c|A|}{d+1}$.
Using Observation \ref{obs:duality} again, we deduce the following:

\begin{theorem}
Let $A$ be an arrangement of hyperplanes in $\mathbb{R}^d$. There is a constant $c=c(d)$ such that there is a query point $q$ with hyperplane enclosing depth $\text{HED}(A,q)\geq \frac{c|A|}{d+1}$.
\end{theorem}

Combining all of the above, we get an analogue to Theorem \ref{thm:super-additive_measures}.

\begin{theorem}
Let $\rho$ be an enclosable depth measure for hyperplanes. Then for all $A\in A^{\mathbb{R}^d}$ and $q\in\mathbb{R}^d$ we have $\text{RD}(A,q)\geq\rho(A,q)\geq\text{HED}(A,q)\geq c\cdot\text{RD}(A,q)$.
\end{theorem}

In particular, all combinatorial depth measures for hyperplanes that we consider in this paper are constant factor approximations of regression depth. In the next three sections, we give three lower bounds for the depth of a deepest point. In Section \ref{sec:Tverberg} we give a lower bound for hyperplane Tverberg depth, in Section \ref{sec:Helly} a slightly stronger bound for regression depth, and in Section \ref{sec:weighted} we give a lower bound for super-additive depth measures with contractible depth regions in the more general setting of weighted arrangements.

\section{A first lower bound: Hyperplane Tverberg Depth}\label{sec:Tverberg}

In this section we prove an analogue of Tverberg's theorem for hyperplane arrangements, resolving a conjecture by Rousseeuw and Hubert from 1999 \cite{Rousseeuw1999}. Our proof is inspired by the proof of Tverberg's theorem by Roudneff \cite{Roudneff}, see also \cite{Tverberg_Survey}.

\tverberg*

\begin{proof}
Let $\pi$ be a partition of $A$ into $r$ parts, each of size at most $d+1$. Note that $\pi$ can have at most $d$ parts of size $\leq d$. Define the following function $f_{\pi}:\mathbb{R}^d\rightarrow\mathbb{R}_{\geq 0}$: for each point $q\in\mathbb{R}^d$, consider the point set $A^*_q$. The partition $\pi$ induces a partition of this point set into parts $X_1(q),\ldots,X_r(q)$. Let $B(q)$ be the smallest ball centered at $q$ which for every part intersects the convex hull, and define $f_{\pi}(q)$ as the radius of this ball.
%Clearly, $B(q)$ is tangent to at least one of the convex hulls.
As the map which for a hyperplane $h$ assigns to a point in $\mathbb{R}^d$ the closest point on $h$ is continuous as a function of $q$, the function $f_{\pi}$ is also continuous. Further, the function goes to infinity along any ray, so it attains a minimum.
Denote by $C(q)$ the set of parts whose convex hulls $B(q)$ is tangent to. By general position, we may assume that $|C(q)|\leq d+1$.

Let now $\pi$ be a partition which minimizes $\min f_{\pi}$ and let $p$ be a point where $f_{\pi}$ attains its minimum. If $f_{\pi}(p)=0$, then by Observation \ref{obs:duality}, $p$ is the desired point.
So, assume that $f_{\pi}(p)>0$. For each $X_j(q)$ let $y_j(q)$ denote the unique point in $\text{conv}X_j(q)$ that minimizes the distance to $q$, i.e., $d(q,\text{conv}X_j)=||q-y_j(q)||$, and define $Y_j(q)\subset X_j(q)$ as the unique subset for which $y_j(q)$ lies in the relative interior of $\text{conv}(Y_j(q))$.
In particular we can write $f_{\pi}(q)=\frac{1}{|C(q)|}\sum_{X_j\in C(q)} ||q-y_j(q)||$, and its gradient as $\nabla f_{\pi}(q)=\frac{1}{|C(q)|}\sum_{X_j\in C(q)} (q-y_j(q))$.
As $f_{\pi}$ is minimized at $p$, we have $\nabla f_{\pi}(p)=0$.

We claim that $C(p)$ consists of exactly $d+1$ parts and that no $d$ of the corresponding vectors $(p-y_j(p))$ lie in a common hyperplane with $p$.
Assume for the sake of contradiction that the latter is not the case, that is, that there is a hyperplane $h$ containing all except possibly one of the vectors $(p-y_j(p))$. Let $\ell$ be a line through $p$ that is orthogonal to $h$. Note that all except possibly one of the affine hulls $\text{aff}Y_j(p)$ for $X_j(p)\in C(p)$ are parallel to $\ell$. If there is a single vector not in $h$, then this vector induces a direction on $\ell$.
Move $p$ a distance $\varepsilon$ in the opposite direction. If all vectors are in $h$, then move $p$ along $\ell$ in any direction. Call the resulting point $p'$. We can choose $\varepsilon$ small enough that $C(p')=C(p)$. Let $h'$ be the hyperplane through $p'$ that is parallel to $h$ and let $h^+$ be its side containing $p$. Consider now the point $y_j(p')$ for some $X_j(p)\in C(p)$. This point is in the relative interior of the points in $Y_j(p')$. Let $a(p')\in Y_j(p')$ and let $a(p)$ be the corresponding point in $Y_j(p)$. If $a(p)$ is on the same side of $h$ as $p'$, then $d(a(p'),p')<d(a(p),p)$ and if $a(p)$ is on the other side then $d(a(p'),p')>d(a(p),p)$, see Figure \ref{fig:tverberg_moving}.
In particular, The affine subspace $Y_j(p')$ is not parallel to $\ell$ and the vector $(p'-y_j(p'))$ points into $h^+$.
As this holds for any $X_j\in C(p)$, then also the gradient $\nabla f_{\pi}(p')=\frac{1}{|C(p)|}\sum_{X_j\in C(p)} (p'-y_j(p'))$ points into $h^+$, and as this is the side that contains $p$, this means that $p$ cannot be a local minimum, which is a contradiction to the choice of $p$.
It follows that any $d$ of the vectors $(p-y_j(p))$ are linearly independent, and thus we need at least $d+1$ of them to have $\nabla f_{\pi}(p)=0$.

Thus, the ball $B(p)$ is tangent to exactly $d+1$ convex hulls, and the $d+1$ tangent hyperplanes form a simplex containing $p$ in its interior.
As there are at most $d$ parts of size $\leq d$, there must be a point $v$ in some $X_j$ such that $B(p)$ still intersects the convex hull of $X_j\setminus\{v\}$. This point must lie on the same side as $p$ of some other tangent hyperplane, say of $X_i$. Then adding $v$ to $X_i$ gives a new partition $\pi'$ in which $B(p)$ intersects the interior of the convex hull of $X_i$. In particular, due to the arguments above, $p$ is not a minimum of $f_{\pi'}$, and thus $\min f_{\pi'}<\min f_{\pi}$. This is a contradiction to the choice of $\pi$, showing that $\min f_{\pi}=0$.
\end{proof}

\begin{figure}
\centering
\includegraphics[scale=0.8]{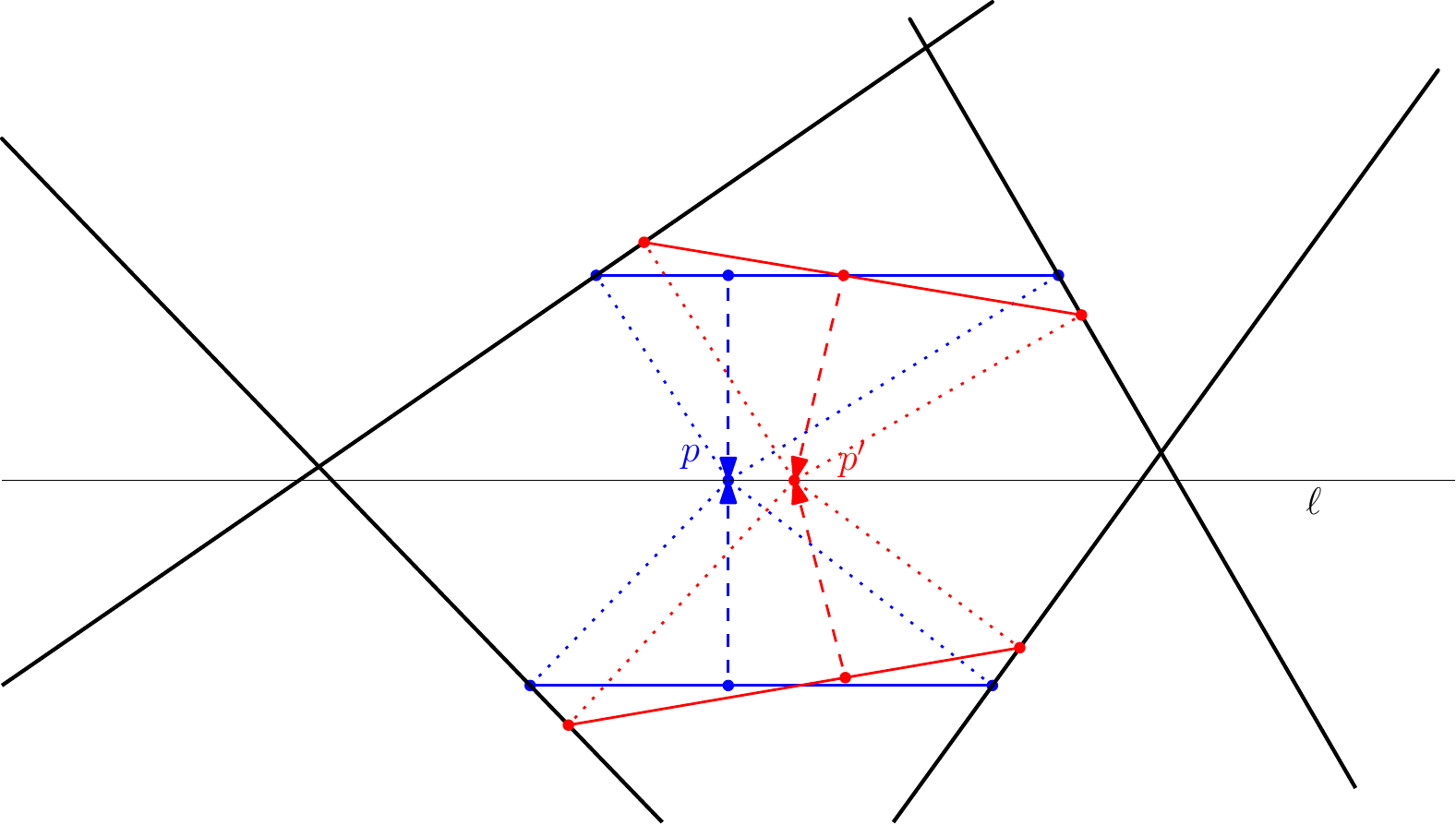}
\caption{Moving $p$ to $p'$ the affine hulls of $Y_j$ are not parallel to $\ell$ anymore.}
\label{fig:tverberg_moving}
\end{figure}

From Theorem \ref{thm:tverberg}, for any super-additive depth measure the existence of a point with depth at least $\frac{|A|}{d+1}$ follows using Theorem \ref{thm:super-additive_measures}.

The existence of a Tverberg theorem for regression depth naturally leads us to ask about a colorful version of such a result.

\begin{conjecture}
Let $r, d$ be a positive integer and $A_1, \dots, A_{d+1}$ be sets of $r$ hyperplanes each in $\mathbb{R}^d$.  There exists a partition of their union into $r$ sets $B_1, \dots B_r$ such that $|A_i \cap B_j| = 1$ for every $i \in [d+1], j\in [r]$ and a point $q$ such that $q$ has positive regression depth for each $B_j$.
\end{conjecture}

In the plane, Karasev conjectured, provided the hyperplanes are in general position, such a partition could be found so that $q$ was in the simplex determined by each $B_j$, since his Tverberg-type results for hyperplanes hold in that setting \cite{Karasev:2008jl}.  However, his conjecture and its natural extensions to $\mathbb{R}^d$ have been disproved \cite{Carvalho2022, Lee2018}.  Yet, those counterexamples do not disprove the regression depth version, in which the $q$ can be in the simplex determined by each $B_j$ or the union of the hyperplanes making $B_j$.

\section{A second lower bound: topological Helly theorem}\label{sec:Helly}

In this section, we give a proof for the centerpoint theorem for regression depth based on one of the first topological versions of Helly's theorem, which states that given a finite family $F$ of subsets of $\mathbb{R}^d$ with the property that for any $d+1$ or fewer of them their intersection is non-empty and contractible, there is a point in the intersection of all families \cite{TopoHelly}. In fact, this method proves a stronger statement: we will show that for an arrangement $A$ in general position, there is always as point in a cell of $A$ of regression depth $\lceil\frac{|A|-d}{d+1}\rceil$. As we will see, this implies that there is always a point of regression depth $\lfloor\frac{|A|}{d+1}\rfloor+1$.

The basic idea is the following: given an arrangement $A$ of hyperplanes, consider some direction $\ell$, and for every point $q$ in $\mathbb{R}^d$ compute how many hyperplanes of $A$ the open ray with direction $\ell$ emanating from $q$ intersects. Denote this number by $\ell(q)$. Define $R_A(k,\ell)$ as the set of points where $\ell(q)\geq k$. As $A$ is finite, there are only finitely many different such regions. If we can show that for $k=\lceil\frac{|A|-d}{d+1}\rceil$ the intersection of any $d+1$ or fewer such regions is contractible, then the existence of a deep point as claimed above follows from the topological Helly theorem.
In fact, our arguments will show that any non-empty depth region is contractible. 

There is however one technical issue: the depth regions of regression depth are in general not contractible. Consider three lines in the plane that form a triangle. The regression depth is 1 on any line or in the interior of the triangle, but it is 2 on the three corners, where two of the lines intersect. So, the region of depth 2 consists of three isolated points and is thus not contractible.

If we however look only at the $2$-dimensional cells of a planar line arrangement, then it is easy to show that the closure of the union of cells of depth at least $k$ is contractible: no cell can be completely surrounded by cells of larger depth, as any ray witnessing depth $k$, that is, intersecting exactly $k$ lines, also witnesses that the other cells it intersects all have depth smaller than $k$.

%Going to higher dimensions, in order to circumvent the issues on the lower-dimensional faces of the hyperplane arrangement, we define a new measure, which we call \emph{truncated regression depth}, denoted by $\text{RD'}$ as follows: let $A$ be an arrangement of hyperplanes in $\mathbb{R}^d$. We first slightly perturb $A$ to get an arrangement $A'$ in general position. In particular, in any $k$-dimensional affine subspace at most $d-k$ of the hyperplanes intersect. Reversing this perturbation induces a surjective map $\pi$ of the faces of $A'$ to the faces of $A$. For any face $F$ of $A$, we call $\pi^{-1}(A)$ the faces \emph{perturbed} from $F$. Note that if $A$ was already in general position, then $\pi$ is a bijection.

To overcome this issue, we define a new measure, which we call \emph{open regression depth}, denoted by $\text{RD'}$ as follows: let $A$ be an arrangement of hyperplanes in $\mathbb{R}^d$. We first slightly perturb $A$ to get an arrangement $A'$ in general position. In particular, in any $k$-dimensional affine subspace at most $d-k$ of the hyperplanes intersect. Reversing this perturbation induces a surjective map $\pi$ of the faces of $A'$ to the faces of $A$. For any face $F$ of $A$, we call $\pi^{-1}(A)$ the faces \emph{perturbed} from $F$. Note that if $A$ was already in general position, then $\pi$ is a bijection.

Consider now the perturbed arrangement $A'$. For any point $q\in\mathbb{R}^d$, define the open regression depth with respect to the perturbed arrangement as the minimum number of hyperplanes of $A$ that any ray emanating from $q$ crosses or is parallel to, where a ray crosses a hyperplane if there is a point in the relative interior of the ray that is also on the hyperplane.
In other words, the open regression depth for perturbed arrangement is just the regression depth, where we do not count the hyperplanes that $q$ lies on. The depth regions of open regression depth in a perturbed arrangement are the unions of cells with large enough depth, with lower-dimensional faces added whenever they are incident to only deep enough cells.

In order to extend the definition to the original arrangement, we define the open regression depth of a query point $q$ in some face $F_q$ of the arrangement $A$ as $\text{RD'}(A,q):=\max_{q'\in F\in F_q}\{\text{RD'}(A',q)\}$, that is, as the maximum open regression depth of any point in one of the faces perturbed from $F_q$. Note that we can perturb the arrangement in a deterministic way, ensuring that the open regression depth is well defined. The following lemma follows immediately from the definition:

\begin{lemma}
For any arrangement of hyperplanes $A$ and any query point $q$, we have $\text{RD'}(A,q)\leq\text{RD}(A,q)$.
\end{lemma}

In particular, proving the existence of deep points for open regression depth implies the existence of deep points for regression depth.
Note, however, that open regression depth is not super-additive: it does not satisfy condition (iii).
We will now prove the existence of deep points for open regression depth using the approach sketched above.
We show that we have the necessary ingredients to apply the topological Helly theorem, starting with the contractability of the relevant regions. Recall that we defined the regions $R_A(k,\ell)$ as the set of points where $\ell(q)\geq k$ for a hyperplane arrangement $A$ and a direction $\ell$, where we considered the relevant ray to be open, that is, not containing $q$. Also recall that as $A$ is finite, is is sufficient to restrict our attention to finitely many directions, and we may assume that these directions are $d$-wise linearly independent, that is, any $d$ of them span a $d$-dimensional cone.

Our proof of contractability requires some algebraic topology, in particular the concept of Mayer-Vietoris sequences in homology theory. We refer to the many excellent books on algebraic topology for the background, e.g.\ \cite{Bredon,Hatcher}.
To show the contractability of some topological space $X$, by the theorems of Whitehead and Hurewicz, it is sufficient to show that $\pi_1(X)=0$, as well as $\Tilde{H}_i(X)=0$ for all $i\geq 0$. In our case, if a loop cannot be contracted, then this has to be because of some family of faces of the arrangement blocking any contraction. But then this family also defines a generator for $H_1$, so we are working with tame enough spaces where vanishing first homology implies simply connectedness. In particular, it suffices to show that all reduced homologies vanish to show that $X$ is contractible. In other words, it is enough to show that $X$ is a \emph{homology cell}.

\begin{lemma}\label{lem:homology_cell}
Let $A_1,\ldots,A_m$ be open subsets of $\mathbb{R}^d$. Assume that each set is a homology cell and that the union of any $d$ of them is a homology cell. Then $\bigcap A:=\bigcap_{i=1}^m A_i$ is either empty or a homology cell.
\end{lemma}

\begin{proof}
Assume that $\bigcap A$ is not empty. We want to show that it is a homology cell.
We prove this statement by induction on the number $m$ of subsets. In fact, we will prove the following stronger statement: for any $1\leq m$, the subset $S_j=(A_1\cap\ldots\cap A_{j-1})\cap(A_j\cup\ldots\cup A_m)$ is a homology cell. We call such a subset a \emph{cap-cup set} For $m=1$ the statement is trivial. For $m=2$ consider the Mayer-Vietoris sequence
\[H_k(A_1\cup A_2)\rightarrow H_{k-1}(A_1\cap A_2)\rightarrow H_{k-1}(A_1)\oplus H_{k-1}(A_2).\]
As $A_1$, $A_2$, as well as $A_1\cup A_2$ are homology cells by assumption, it follows that $A_1\cap A_2$ is also a homology cell.

For the general case of a subset $S_j=(A_1\cap\ldots\cap A_{j-1})\cap(A_j\cup\ldots\cup A_m)$, write $A:=(A_1\cap\ldots\cap A_{j-1})$, $B:=A_j$ and $C:=(A_{j+1}\cup\ldots\cup A_m)$. Note that $S_j=A\cap(B\cup C)=(A\cap B)\cup(A\cap C)$. Further note that $A\cap B\cap C=S_{j+1}$ and $A\cup B\cup C=S_{j-1}$.

Consider first the Mayer-Vietoris sequence for the two sets $A\cap B$ and $A\cap C$:
\[H_k(A\cap B)\oplus H_k(A\cap C)\rightarrow H_k((A\cap B)\cup(A\cap C))\rightarrow H_{k-1}(A\cap B\cap C)\rightarrow H_{k-1}(A\cap B)\oplus H_{k-1}(A\cap C).\]
Assuming that $A\cap B$ and $A\cap C$, both of which are cap-cup sets, are homology cells, it would follow that $H_k(S_j)\simeq H_{k-1}(S_{j+1})$.
Consider now the Mayer-Vietoris sequence for the sets $A$ and $B\cup C$:
\[H_{k+1}(A)\oplus H_{k+1}(B\cup C)\rightarrow H_{k+1}(A\cup B\cup C)\rightarrow H_k(A\cap(B\cup C))\rightarrow H_{k}(A)\oplus H_{k}(B\cup C).\]
Assuming that $A$ and $B\cup C$, both of which are cap-cup sets, are homology cells, it would follow that $H_{k+1}(S_{j-1})\simeq H_{k}(S_{j})$.
In particular, assuming that all cap-cup sets defined by fewer than $m$ sets are homology cells, we get a chain of isomorphisms
\[H_k\left(\bigcap A\right)=H_k(S_m)\simeq H_{k+1}(S_{m-1})\simeq\ldots\simeq H_{k+j}(S_{m-j})\simeq\ldots\simeq H_{k+m-1}\left(\bigcup A\right).\]
As the union of any $d$ sets is a homology cell, for $m\leq d$ we get the claimed result by induction.
Consider now the case $m=d+1$. As we consider subsets of $\mathbb{R}^d$, we have $H_j(\bigcup A)=0$ for $j\geq d$.
From the chain of isomorphisms we get that $H_k(\bigcap A)\simeq H_{k+d}(\bigcup A)$, and thus $H_k(\bigcap A)=0$ for $k\geq 0$.
By assumption $\bigcap A$ is not empty, so this implies that $\bigcap A$ has the homology of a point, that is, $H_k(\bigcap A)=0$ for all $k\in\mathbb{Z}$. In particular, $\bigcap A$ is a homology cell and thus, by the chain of isomorphisms, so are all cap-cup sets defined by $d+1$ subsets.
The statement for $m>d+1$ now again follows inductively.
\end{proof}

We can apply this result to our setting.

\begin{lemma}\label{lem:contractible}
Let $\ell_1,\ldots,\ell_m$ be directions in $d$-wise general position in $\mathbb{R}^d$, let $A$ be a hyperplane arrangement and let $R(k):=\bigcap_{i=i}^m R_A(k,\ell)$. If $R(k)\neq\emptyset$ then $R(k)$ is contractible.
\end{lemma}

In particular, the depth regions, that is, the intersections of $R_A(k,\ell)$ over all considered directions $\ell$ is contractible.

\begin{proof}
As in our setting homology cells are contractible, by Lemma \ref{lem:homology_cell} it suffices to show that the union of any $d$ regions $R_A(k,\ell_1),\ldots ,R_A(k,\ell_d)$ is contractible. Denote this union by $U(k)$ and let $n=|A|$. Let $C\subsetneq \mathbb{R}^d$ be the cone spanned by the $d$ directions and let $\ell_0$ be a direction in $-C$. In particular, moving from any point in $U(k)$ in direction $\ell_0$ we never leave $U(k)$. Thus, $U(k)$ is contractible.
\end{proof}

The final property that we need in order to apply the topological Helly theorem is that the intersection of any $d+1$ or fewer regions is non-empty.

\begin{lemma}\label{lem:intersections}
Let $\ell_1,\ldots,\ell_m$ be directions in $\mathbb{R}^d$, $m\leq d+1$, let $A$ be a hyperplane arrangement and let $R:=\bigcap_{i=i}^m R_A(k,\ell)$ for $k\leq\lceil\frac{|A|-d}{d+1}\rceil$. Then $R\neq\emptyset$.
\end{lemma}

\begin{proof}
For every direction $\ell_i$ let $h_i$ be a hyperplane orthogonal to $\ell_i$ which bounds a half-space $h_i^+$ that contains $R_A(k,\ell_i)$. In paticular, for any point $p$ in $h_i^+$, moving $p$ in direction $\ell_i$, we eventually enter $R_A(k,\ell_i)$ and never leave it again.
Thus, if all these half-spaces have a common intersection, then this intersection can be translated to lie in $R$, showing $R\neq\emptyset$.
So, assume that these half-spaces have an empty intersection. As we assumed that any $d$ of our directions are linearly independent, this can only happen for $m=d+1$. In this case, we find a point $q\in \mathbb{R}^d$ such that the $d+1$ (closed) rays emanating from $q$ with directions $-\ell_i$ all intersect strictly more than $|A|-k$ hyperplanes of $A$. Each hyperplane that does not contain $q$ can intersect at most $d$ of the rays, and by the general position assumption, at most $d$ hyperplanes contain $q$. Thus, if $x$ denotes the number of intersections between rays and hyperplanes, we have
\[(|A|-k)(d+1)<x\leq (|A|-d)d+d(d+1)=(|A|+1)d.\]
Rearranging this and using that all numbers are integers gives $k>\lceil\frac{|A|-d}{d+1}\rceil$, which is a contradiction to the assumption, showing that $R\neq\emptyset$.
\end{proof}

Now we have all the ingredients that are necessary for the topological Helly theorem, and we deduce the following

\begin{corollary}
For every hyperplane arrangement $A$ in $\mathbb{R}^d$ there is a point $q\in\mathbb{R}^d$ for which the open regression depth is $\text{RD'}(A,q)\geq \lceil\frac{|A|-d}{d+1}\rceil$.
\end{corollary}

In particular, by the definition of open regression depth, if the arrangement $A$ is in general position, such a point can be chosen in a cell of $A$.
It remains to show that we can get even deeper points for regression depth.

\begin{lemma}\label{lem:open_to_regression}
For every hyperplane arrangement $A$ in $\mathbb{R}^d$ there is a point $q\in\mathbb{R}^d$ for which the regression depth is $\text{RD}(A,q)\geq \lfloor\frac{|A|}{d+1}\rfloor+1$.
\end{lemma}

\begin{proof}
Consider a point $q$ in a cell $C$ of maximum open regression depth $k$, and let $\partial C$ be the boundary of the cell $C$. If there is a point on $\partial C$ with open regression depth $k$, then this point has regression depth $k+1$, and the claim follows. So assume that the open regression depth is strictly smaller everywhere on $\partial C$. Then we again find $d+1$ directions such that the rays emanating from $q$ with these directions intersect exactly $k$ hyperplanes. Looking at the opposite directions, the rays thus intersect exactly $|A|-k$ hyperplanes, and as $q$ lies in the interior of a cell every hyperplane intersects at most $d$ rays. Analogous to the proof of Lemma \ref{lem:intersections} we thus get $k>\lceil\frac{|A|}{d+1}\rceil$. This proves the claim for all cases where $d+1$ does not divide $|A|$. If $d+1$ divides $|A|$, note that as soon as one of the hyperplanes only intersects $d-1$ of the considered rays, then we get $k>\lceil\frac{|A|+1}{d+1}\rceil$, and the claim follows again. So, assume that each hyperplane intersects exactly $d$ rays. This gives a partition of the set of hyperplanes into $d+1$ parts, each of size $\frac{|A|}{d+1}$ defined by the ray they do not intersects. The boundary $\partial C$ inherits this partition, and each of the parts is contractible. In particular, $\partial C$ contains a vertex $q$ that is the intersection of $d$ hyperplanes of $d$ different parts. Now every ray emanating from $q$ must intersect all hyperplanes of some part, but also lies on at least $d-1$ other hyperplanes, showing that the regression depth of $q$ is at least $k+d-1$, which is a contradiction to the assumption that the open regression depth is strictly smaller everywhere on $\partial C$.
\end{proof}

Using the above insights, we can also conclude the contractability of many regions of regression depth.

\begin{lemma}\label{lem:centerpoint_region_contractible}
Let $k\leq\lceil\frac{|A|}{d+1}\rceil$. Then the region $R$ of points $p$ whose regression depth is $\text{RD}(A,q)\geq k$ is contractible.
\end{lemma}

\begin{proof}
If there is a point with open regression depth $k$, then $R$ is just the closure of the region of points with open regression depth at least $k$, which is contractible by Lemma \ref{lem:contractible}. Otherwise, by the proof of Lemma \ref{lem:open_to_regression}, $R$ is the union of faces incident only to cells of maximum open regression depth. As no cell is completely surrounded by deeper faces there is a contraction from a cell of maximum open regression depth to the deeper faces incident to it. Thus, as the region of maximum open regression depth is contractible, so is $R$.
\end{proof}

\section{A third lower bound: weighted arrangements}\label{sec:weighted}

In this section we give yet another proof for the existence of points with large regression depth. The proof we give here works for (and actually requires) the more general case of \emph{weighted arrangements} of hyperplanes. A weighted arrangement of hyperplanes is a tuple $(A,w)$ consisting of a finite arrangement $A$ of hyperplanes and a weight function $w:A\rightarrow\mathbb{R}_{\geq 0}$ which assigns to each hyperplane a weight. By a slight abuse of notation we will often just write $A$ for a weighted arrangement. For a subarrangement $A'\subseteq A$ we have $w'(h)\leq w(h)$, where $w'$ is the weight function on $A'$, and we write $w'(A'):=\sum_{h\in A'}w'(h)$. We say that $A'\subsetneq A$ is a strict subset of $A$ if the underlying hyperplane arrangement of $A'$ is a strict subset of that of $A$. The definition of regression depth extends to weighted arrangements: for any ray $r$ emanating from a query point $q$, let $A(r)$ be the hyperplanes intersected by $r$. Then, the regression depth $\text{RD}(A,q)$ of $q$ is the minimum of $w(A(r))$ taken over all rays emanating from $q$.
This definition is similar to, but more restrictive than a measure-theoretic generalization of regression depth considered by Mizera \cite{Mizera}.

Our proof also works for more general families of depth measures on weighted hyperplane arrangements. We extend the definition of super-additive depth measures above to weighted hyperplane arrangements as follows:

\begin{enumerate}
    \item[(i)] for all $A\in A^{\mathbb{R}^d}$ and $q\in\mathbb{R}^d$ and any hyperplane $h$ we have $|\rho(A,q)-\rho(A\cup\{h\},q)|\leq w(h)$,
    \item[(ii)] for all $A\in A^{\mathbb{R}^d}$ we have $\rho(A,q)=0$ if $q$ is in an unbounded cell of $A$,
    \item[(iii)] for all $A\in A^{\mathbb{R}^d}$ we have $\rho(A,q)\geq \min\{w(h)\mid h\in A\}$ if $q$ is in a bounded cell or if $q$ lies on a hyperplane of $A$,
    \item[(iv)] for any disjoint subsets $A_1,A_2\subseteq A$ and $q\in\mathbb{R}^d$ we have $\rho(A,q)\geq \rho(A_1,q)+\rho(A_2,q)$.
\end{enumerate}

Note that any hyperplane arrangement can be considered as a weighted hyperplane arrangement by assigning weight 1 to each hyperplane. On the other hand, each depth measure for hyperplane arrangement can be extended to a depth measure on weighted hyperplanes: using the fact that $\mathbb{Q}$ is dense in $\mathbb{R}$, we can place multiple hyperplanes in the same position and the normalize to get a weighted arrangement.

For a weighted arrangement of hyperplanes $A$ and a depth measure $\rho$ denote by $R_{\rho}^A(\alpha):=\{q\in\mathbb{R}^d\mid \rho(A,q)\geq \alpha\}$ the $\alpha$-\emph{depth region}. The \emph{median region}, which is the deepest non-empty depth region, is denoted by $M_{\rho}^A$.

\begin{theorem}\label{thm:weighted_centerpoint}
Let $A$ be a weighted arrangement of hyperplanes in $\mathbb{R}^d$ and let $\rho$ be a super-additive depth measure on weighted hyperplanes whose depth regions are compact and contractible. Then there exists a point $q\in\mathbb{R}^d$ for which $\rho(A,q)\geq \frac{w(A)}{d+1}$.
\end{theorem}

Before we prove Theorem \ref{thm:weighted_centerpoint}, we give some lemmata that we will need in the proof.
The first lemma concerns a generalization of a section in a vector bundle. Let $\pi: E\rightarrow B$ be a real vector bundle over a compact manifold $B$. Following \cite{Zivaljevic} we say that $\phi:B\rightarrow E$ is a \emph{multisection} if for every $x\in B$ we have that $\phi(x)\subseteq F_x:=\pi^{-1}(x)$. We further say that $\phi$ is \emph{contractible} if it is contractible in each fiber, that is, for every $x\in B$ the set $\phi(x)$ is contractible. Finally, we say that $\phi$ is compact if $\Gamma(\phi):=\{(x,v)\mid v\in\phi(x)\}\subseteq B\times E$ is compact.
For any multisection $\phi$, denote by $Z(\phi)$ its intersection with the zero section.

\begin{lemma}\label{lem:multisections}
Let $\pi: E\rightarrow B$ be a real vector bundle over a compact manifold $B$. Let $\phi$ be a compact contractible multisection. Then there is a section $s$ with $Z(s)=Z(\phi)$. In particular, if $\pi$ has no nowhere zero section, then $\phi$ must intersect the zero section.
\end{lemma}

\begin{proof}
Let $B_Z:=\{x\in B\mid 0 \in F_x, \ 0\in \phi(x)\}$ be the (closed) region of $B$ where $\phi$ intersects the zero section. Set $s(x)=0\in F_x$ for all $x\in B_Z$.
Let $B_C$ be the closure of the complement of $B_Z$. Consider a fine enough triangulation $T$ of $B_C$. For every vertex $v$ of $T$, pick some point $p_v\in\phi(v)$ and set $s(v):=p_v$. Consider now some face $F$ of $T$ on which $s$ is not defined but on whose boundary $s$ is defined. By a Vietoris--Begle type theorem due to Smale (\cite{Smale}, "Main Theorem"), $\phi(F)$ is homotopy equivalent to $F$. As $F$ is a simplex, $\phi(F)$ thus does not have any non-trivial homotopy groups, so there is no obstruction to continuously extending $s$ on $F$, giving a continuous function $s_F:F\rightarrow\pi^{-1}(F)$. This map is not necessarily a section, as it might map points of $F$ outside their fiber. However, for each $x\in F$ the fibers $F_x$ are isomorphic, thus by mapping $s_F(x)\in F_y$ to the corresponding point in the correct fiber $F_x$, we get an extension of the section $s$ to $F$. Having extended $s$ to all of $B$, we get a section $s$ for which we have by construction that $Z(s)=Z(\phi)$.
\end{proof}

The second lemma is about partitions of hyperplane arrangements.

\begin{lemma}\label{lem:partition}
Let $\rho$ be a depth measure for weighted hyperplanes whose depth regions are compact and contractible and let $A$ be a weighted hyperplane arrangement in $\mathbb{R}^d$ with $|A|\geq d+2$. Then there exists a partition of $A$ into strict subarrangements $A_1$ and $A_2$ whose median regions intersect.
\end{lemma}

The proof is analogous to the proof of Lemma 9 in \cite{Enclosing}, replacing Proposition 1 from \cite{Zivaljevic} with our Lemma \ref{lem:multisections}.

\begin{proof}
Let $\partial\Delta$ be the boundary of the simplex with vertices $A$ and let $B$ be its barycentric subdivision. There is a natural identification of the vertices of $B$ with strict subsets of $A$, meaning that for any such vertex $b$ we get a strict subarrangement $A(b)$ with $w_b(h)=w(h)$ if $h$ is in the corresponding subsets and $w_b(h)=0$ otherwise. Extending this assignment linearly to $\partial\Delta$, we get a continuous map which assigns to each point $b\in\partial\Delta$ a strict subarrangement $A(b)$. Further, under the standard antipodality on $\partial\Delta$ we get complements of the weighted subarrangements.

Let $M(b)$ denote the median region of the weighted arrangement $A(b)$. We claim that there is a point $b\in\partial\Delta$ for which the median regions $M(b)$ and $M(-b)$ intersect. This proves the claim by setting $A_1:=A(b)$ and $A_2:=A(-b)$.
In order to show the claim, consider the vector bundle $\zeta$ obtained from attaching $\mathbb{R}^d$ to each point of $\partial\Delta$ and taking the quotient with respect to antipodality.
Note that $M(b)$ defines a multisection in $\zeta$, which by Lemma \ref{lem:contractible} is compact and contractible.
Define the negative multisection $-M$ by reflecting $M(b)$ at the origin for each $b$.
For each $b\in\partial\Delta$ consider $Q(b):=M(b)-M(-b)$, defined by taking the Minkowski sum of $M(b)$ and $-M(-b)$.
As Minkowski sums of compact and contractible sets are again compact and contractible, $Q$ is again a compact contractible multisection.
In particular, by Lemma \ref{lem:multisections}, there is a section $s$ whose zeroes coincide with the zeroes of $Q$.
As $\partial\Delta$ is homeomerphic to the sphere $S^{|A|-2}$ and $|A|\geq d+2$, it follows from the Borsuk-Ulam theorem that $s$ and therefore $Q$ has a zero. Thus, there is a point $b$ for which $M(b)$ and $M(-b)$ intersect, as claimed.
\end{proof}

We are now ready to prove Theorem \ref{thm:weighted_centerpoint}.

\begin{proof}[Proof of Theorem \ref{thm:weighted_centerpoint}]
Let $A$ be a weighted arrangement of hyperplanes in $\mathbb{R}^d$.
We prove the statement by induction on the number of hyperplanes in $A$.
If $A$ consists of at most $d+1$ hyperplanes, it follows from condition (iii) that $\rho(A,q)\geq \frac{w(A)}{d+1}$ for some $q\in \mathbb{R}^d$: just take $q$ as any point on a hyperplane of maximum weight.
So assume that $A$ consists of at least $d+2$ hyperplanes.
By assumption the depth regions are compact and contractible.
Thus, by Lemma \ref{lem:partition}, we can partition $A$ into strict subarrangements $A_1$ and $A_2$ whose median regions intersect.
As both $A_1$ and $A_2$ are strict subarrangements, by the induction hypothesis for any point $q$ in the intersection of their median regions we have $\rho(A_1,q)\geq \frac{w_1(A_1)}{d+1}$ and $\rho(A_2,q)\geq \frac{w_2(A_2)}{d+1}$.
As $\rho$ satisfies condition (iv), for any such point we thus have
\[\rho(A,q)\geq \rho(A_1,q)+\rho(A_2,q)\geq \frac{w_1(A_1)+w_2(A_2)}{d+1}=\frac{w(A)}{d+1}.\]
\end{proof}

%Theorem \ref{thm:weighted_centerpoint} again gives a proof of the existence of centerpoints for Regression depth, but we have to be a bit careful in the choice of depth regions. Taking the depth regions as the closures of the depth regions of open regression depth, all the arguments from Section \ref{sec:Helly} go through for the weighted setting and show that the depth regions are compact and contractible, as long as there are at least $d$ hyperplanes in the arrangement. For fewer hyperplanes we can just define the regions in the natural way as the intersections of the hyperplanes. By intersecting them with a large enough ball, these regions can be made compact.

At this point, it is not clear how we can use Theorem \ref{thm:weighted_centerpoint} to prove the existence of centerpoints for regression depth. If we look at the depth regions of regression depth, we have seen in Section \ref{sec:Helly} that they are in general not contractible. To overcome this issue, we have introduced open regression depth and argued that the depth regions of open regression depth are contractible, and these arguments go through even if the arrangement is weighted. However, for a hyperplane arrangement in general position, these regions are by definition open, and thus not compact. Further, open regression depth is not a super-additive depth measure, as it does not satisfy condition (iii). In particular, if $A$ consists of a single hyperplane, then the open regression depth is 0 everywhere, and so the base case of the proof of Theorem \ref{thm:weighted_centerpoint} fails. However, as we have seen in Lemma \ref{lem:centerpoint_region_contractible}, if $k\leq\lceil\frac{|A|}{d+1}\rceil$ the region of regression depth at least $k$ is contractible. Again, the involved arguments go through if the arrangement is weighted, implying that if $k\leq\frac{w(A)}{d+1}$, then the region of regression depth at least $k$ is contractible. Thus, defining a new measure \emph{truncated regression depth} by
\[\text{TRD}(A,q):=\min\left(\frac{w(A)}{d+1},\text{RD}(A,q)\right),\]
we get a measure whose depth regions are closed and contractible.
Clearly, the only unbounded regions are the ones containing an unbounded face of the arrangement, and we can make those compact by intersecting with a sufficiently large ball.
Finally, as regression depth is super-additive, so is truncated regression depth, and by definition, truncated regression depth is bounded from above by regression depth. We thus have the following:

\begin{lemma}\label{lem:truncated}
Truncated regression depth is a super-additive depth measure for hyperplane arrangements which has compact and contractible depth regions. Further, for every arrangement $A$ and every point $q$ we have $\text{TRD}(A,q)\leq\text{RD}(A,q)$.
\end{lemma}

It now follows from Theorem \ref{thm:weighted_centerpoint} that there is always a point of truncated regression depth $\text{TRD}(A,q)\geq \frac{w(A)}{d+1}$ and such a point also has regression depth $\text{RD}(A,q)\geq \frac{w(A)}{d+1}$.

\section{A regression depth version of the center transversal theorem}\label{sec:center-transversal}

Let $A$ be an arrangement of hyperplanes in $\mathbb{R}^d$. Assume that the origin is not contained in any hyperplane in $A$.
Let $L$ be a $k$-dimensional linear subspace of $\mathbb{R}^d$. Then $A\cap L$ is a hyperplane arrangement in $L$.
In particular, we can again study the depth of points $q\in L$ within the Euclidean space $L$ with respect to the arrangement $A\cap L$.
Note however that $A\cap L$ might have smaller cardinality than $A$, as some hyperplanes of $A$ might be parallel to $L$.
In fact, if all of them are parallel to $L$, then $A\cap L$ is empty.
We define the regression depth of $q\in L$ with respect to $A\cap L$ as the minimum number of hyperplanes in $A$ intersected by or parallel to any ray in $L$ emanating from $q$, and denote it by $\text{RD}(A,q,L)$.
In particular, if all hyperplanes in $A$ are parallel to $L$, then $\text{RD}(A,q,L)=|A|$ for all $q\in L$.
This definition extends to open regression depth and truncated regression depth, where we truncate at $\frac{|A_i|}{k+1}$.

\begin{theorem}\label{thm:center-transversal}
Let $1\le k\le d$ be integers and $A_1,\ldots, A_{d-k+1}$ be $d-k+1$ finite arrangements of hyperplanes in $\mathbb{R}^d$.
Then there exists a $k$-dimensional linear subspace $L$ and a point $q\in L$ such that $q$ has regression depth $\text{RD}(A_i,q,L)\geq \frac{|A_i|}{k+1}$ in $L$ for every $i\in\{1,\ldots, d-k+1\}$.
\end{theorem}

\begin{proof}
We will prove the statement for truncated regression depth, which will imply the theorem as regression depth is bounded from below by truncated regression depth.
Consider the Grassmann manifold $Gr_k(\mathbb{R}^d)$ of all $k$-dimensional subspaces of $\mathbb{R}^d$. Let $\gamma_k^d$ be the canonical bundle over $Gr_k(\mathbb{R}^d)$, which has total space $E:=\{(L,v)\mid v\in L\}$ and whose projection $\pi:E\rightarrow Gr_k(\mathbb{R}^d)$ is given by $\pi((L,v))=L$.
For an arrangement $A_i$, let $R_i(L)$ be the set of points in $L$ that have large depth, that is, $R_i(L):=\{v\in L\mid \text{TRD}(A_i,v,L)\geq\frac{|A_i|}{k+1} \}$.
By Lemma \ref{lem:truncated}, each $R_i(L)$ is compact and contractible.
Further, when a hyperplane $h\in A_i$ becomes parallel to $L$, the depth of any point can only increase, thus $R_i(L')\subseteq R_i(L)$ for any $L'$ in a small neighborhood of $L$.
Thus, $R_i$ is a compact contractible multisection.
Define the negative multisection $-R_i$ by reflecting $R_i(L)$ at the origin for each $L$, and for each $i\in\{1,\ldots,d-k\}$ consider $Q_i:=R_{d-k+1}-R_i$, defined by taking the Minkowski sum of $R_{d-k+1}(L)$ and $-R_i(L)$ on each $L$.
As Minkowski sums of compact and contractible sets are again compact and contractible, $Q_i$ is again a compact contractible multisection.
In particular, by Lemma \ref{lem:multisections}, there are sections $s_i$ whose zeroes coincide with the zeroes of $Q_i$.
It was shown in \cite{Zivaljevic}, Prop. 2 (see also \cite{Dolnikov}, Lem. 1), that any $d-k$ sections on $\gamma_k^d$ must have a common zero, that is, there is a subspace $L$ such that $s_1(L)=\ldots s_{d-k}(L)=0$.
By the definition of the sections $s_i$, this implies that there is a point $q\in L$ such that $q\in R_i(L)$ for all $i\in\{1,\ldots,d-k+1\}$.
In particular, $\text{TRD}(A_i,q,L)\geq \frac{|A_i|}{k+1}$ in $L$ for every $i\in\{1,\ldots, d-k+1\}$.
\end{proof}

Since there is a regression depth version of the center transversal theorem and of Tverberg's theorem, a natural question is if there is a generalization of both.  This is still open in the case of finite families of points, since it was conjectured by Tverberg and Vre\'cica in 1993 \cite{Tverberg:1993ia}.

\begin{conjecture}
Let $1 \le k \le d$ be integers and $A_1, \dots, A_{d-k+1}$ be $d-k+1$ finite arrangements of hyperplanes in $\mathbb{R}^d$.  Assume that $|A_i| = (k+1)(r_i-1)+1$ for some positive integer $r_i$, for each $i=1,\dots,d-k+1$.  Then, there exists a $k$-dimensional subspace $L$, a point $q \in L$, and a partition of each $A_i$ into $r_i$ parts $A_i^{(1)}, \dots, A_i^{(r_i)}$ such that
\[
RD(A_i^{(j)}, q, L) \ge 1
\]
for each $i=1,\dots, d-k+1$, $j = 1, \dots, r_i$.
\end{conjecture}

The classic conjecture for families of points, which has similar parameters, has only been confirmed when all $r_i$ are powers of the same prime $p$ and $pk$ is even \cite{Karasev:2007ib}.

\bibliographystyle{plainurl} %
\bibliography{refs, refref}

\end{document}